\documentclass[usenames,dvipsnames,conference]{IEEEtran} 
\usepackage{cite}
\usepackage{amsmath,amssymb,amsfonts}
\usepackage{algorithmic}
\usepackage{graphicx}
\usepackage{textcomp}
\usepackage[usenames,dvipsnames,pdftex]{xcolor} 

\usepackage{calc}
\usepackage{textpos}

\usepackage{graphicx}
\usepackage{balance}  

\usepackage{pifont}

\newcommand{\annote}[3]{{%
		\colorbox{#3}{\bfseries\sffamily\footnotesize\textcolor{white}{#2}}%
		\color{#3}%
		\ifthenelse{\equal{#1}{}}{}{%
      $\blacktriangleright$\textit{#1}$\blacktriangleleft$}%
    }%
}

\newcommand{\todo}[1]{\annote{#1}{Todo}{red}}

\newcommand{\commentFT}[1]{\annote{#1}{FT}{magenta}}

\newcommand{\commentOR}[1]{\annote{#1}{OR}{OliveGreen}}

\renewcommand{\annote}[3]{}

\newcommand{\mlone}{ml1M\xspace}
\newcommand{\MLone}{MovieLens1M\xspace}
\newcommand{\mlten}{ml10M\xspace}
\newcommand{\MLten}{MovieLens10M\xspace}
\newcommand{\mltwenty}{ml20M\xspace}
\newcommand{\MLtwenty}{MovieLens20M\xspace}
\newcommand{\am}{AM\xspace}
\newcommand{\AM}{AmazonMovies\xspace}
\newcommand{\dblp}{DBLP\xspace}
\newcommand{\DBLP}{DBLP\xspace}
\newcommand{\gw}{GW\xspace}
\newcommand{\GW}{Gowalla\xspace}

\newcommand{\BruteForce}{Brute Force\xspace}
\newcommand{\Hyrec}{Hyrec\xspace}
\newcommand{\NNDescent}{NNDescent\xspace}
\newcommand{\LSH}{LSH\xspace}

\newcommand{\GF}{GoldFinger\xspace}

\newcommand{\HyperLSH}{Cluster-and-Conquer\xspace}
\newcommand{\hyperLSH}{Cluster-and-Conquer\xspace}
\newcommand{\HYPERLSH}{Cluster-and-Conquer\xspace}
\newcommand{\HLSH}{$C^2$\xspace}

\newcommand{\MinHash}{MinHash\xspace}

\newcommand{\profileSim}{f_{\mathit{sim}}}

\newcommand{\aknn}[1]{\widehat{\mathit{knn}}(#1)}

\newcommand{\agknn}{\widehat{G}_{\operatorname{KNN}}}
\newcommand{\gknn}{G_{\operatorname{KNN}}}
\newcommand{\itemsOfU}[1]{\mathit{P}_{#1}}

\newcommand{\qualityFunc}{\mathit{avg\_sim}}
\newcommand{\qualityRatioFunc}{\mathit{quality}}

\newcommand{\dataset}[1]{\multirow{4}{*}{\begin{sideways}#1\end{sideways}}}
\newcommand{\vertbrace}[4]{\raisebox{#4}{\multirow{#1}{*}{\begin{sideways}$\begin{array}{c}#2\\\overbrace{\hspace{#3}}\end{array}$\end{sideways}}}}
\newcommand{\datasetsep}{\vspace{1em}}
\newcommand{\myHrule}{\cline{3-5}}
\renewcommand{\datasetsep}{\myHrule\\[-0.5em] \myHrule}
\newcommand{\mycolsep}{\hspace{0.5em}}

\newcommand{\knnFunc}[1]{knn(#1)}

\newcommand{\dictionariesOrSets}{sets\xspace}

\definecolor{Gray}{gray}{0.88}

\newcommand{\FMH}{Fast\-Random\-Hash\xspace}

\newcommand{\nbHash}{t\xspace}
\newcommand{\nbHashTry}{N\xspace}
\newcommand{\nbBuckets}{b\xspace}

\makeatletter
\newcommand*{\shifttext}[2]{%
  \settowidth{\@tempdima}{#2}%
  \makebox[\@tempdima]{\hspace*{#1}#2}%
}
\makeatother

\newcommand{\Pin}{P_{\cap}} 
\newcommand{\prob}[1]{\mathbb{P} [#1]}

\newcommand{\high}{\cellcolor{Gray!75}}

\newcommand{\nbColl}{\kappa}
\newcommand{\Cu}{\ell} 
\newcommand{\probNaked}{\mathbb{P}}
\newcommand{\expect}[1]{\mathbb{E} \left[#1\right]}

\usepackage{amsmath}
\usepackage{amssymb}
\usepackage{stmaryrd}

\usepackage{dsfont} 

\usepackage{booktabs} 
\usepackage{colortbl} 
\usepackage{hhline}

\usepackage{relsize}
\usepackage[caption=false]{subfig}

\usepackage[labelfont=bf]{caption}

\usepackage[ruled,vlined]{algorithm2e}

\usepackage{listings}

\definecolor{dkgreen}{rgb}{0,0.6,0}
\definecolor{gray}{rgb}{0.5,0.5,0.5}
\definecolor{mauve}{rgb}{0.58,0,0.82}

\usepackage{xspace}

\usepackage{epsfig, floatflt,amssymb} 
\usepackage{moreverb} 

\usepackage{multirow} 
\usepackage{rotating}
\usepackage{url} 
\usepackage[all]{xy} 
\usepackage{shorttoc} 

\usepackage[normalem]{ulem} 

\usepackage[export]{adjustbox}
\usepackage{stackengine}

\usepackage{amsthm}

\newtheorem{theorem}{Theorem}

\begin{document}
\title{\HyperLSH:\\When Randomness Meets Graph Locality} 

\author{%
\IEEEauthorblockN{George Giakkoupis}
\IEEEauthorblockA{\textit{Univ Rennes, Inria, CNRS, IRISA} \\
Rennes, France \\
george.giakkoupis@inria.fr}
\and
\IEEEauthorblockN{Anne-Marie Kermarrec}
\IEEEauthorblockA{\textit{EPFL} \\
Lausanne, Switzerland \\
anne-marie.kermarrec@epfl.ch}
\and
\IEEEauthorblockN{Olivier Ruas}
\IEEEauthorblockA{\textit{Inria, Univ. Lille} \\
  Lille, France \\
olivier.ruas@inria.fr}
\and
\IEEEauthorblockN{Fran\c{c}ois Ta\"iani\textsuperscript{\textdagger}}
\IEEEauthorblockA{\textit{Univ Rennes, Inria, CNRS, IRISA} \\
Rennes, France \\
francois.taiani@irisa.fr}
}


\maketitle

\begin{abstract}
  K-Nearest-Neighbors (KNN) graphs are central to many emblematic data mining and machine-learning applications.
  Some of the most efficient KNN graph algorithms are incremental and local: they start from a random graph, which they incrementally improve by traversing neigh\-bors-of-neighbors links.
  Paradoxically, this  random start is also one of the key weaknesses of these algorithms:
  nodes are initially connected to dissimilar neighbors\commentOR{"unsimilar" is not proper english: "un-similar"? "highly different nodes"?}\commentFT{using dissimilar}, that lie far away according to the similarity metric. As a result, incremental algorithms must first laboriously explore spurious potential neighbors before they can identify similar nodes, and start converging.\commentFT{looping back to incremental algorithms}
  In this paper, we remove this drawback with \HyperLSH (\HLSH for short). 
  \todo{Highlight what is unique about \HLSH. Sounds like tread of the mill.}
   \HyperLSH boosts the starting configuration of greedy algorithms thanks to a novel lightweight clustering mechanism, dubbed \FMH. \FMH leverages randomness and recursion to pre-cluster similar nodes at a very low cost.
  Our extensive evaluation on real datasets shows that \HyperLSH significantly outperforms existing approaches, including LSH, yielding speed-ups of up to $\times 4.42$ while incurring only a negligible loss in terms of KNN quality.
  \commentOR{Percent instead of Speed-ups?}
\end{abstract}



\begin{IEEEkeywords}
  KNN graph, Big Data
\end{IEEEkeywords}

\IEEEpeerreviewmaketitle 

{\renewcommand{\thefootnote}{\textdagger}\footnotetext{Authors are listed in alphabetical order.}}

\section{Introduction}
\label{chap:HyperLSH:intro}

$k$-Nearest-Neighbors (KNN) graphs\footnote{Note that the problem of computing a complete KNN graph (which we address in this paper) is related but different from that of answering a sequence of KNN queries.} play a fundamental role in many
emblematic data-mining and machine-learning applications, including classification~\cite{gorai2011employing,DBLP:journals/corr/NodarakisSTTP14}, recommender systems~\cite{Hyrec,Levandoski:2012:LLR:2310257.2310356,linden2003amazon,sarwar2001item,campos2010simple}, dimensionality reduction~\cite{chen2009fast} and graph signal processing~\cite{Tremblay:SignalProcessing2016}.
A KNN graph connects each node of a dataset to its $k$ closest counterparts (its {\it neighbors}), according to some application-dependent 
similarity metric
.
In many applications, this similarity is computed from a second set of entities, termed \emph{items}, associated with each node. (For instance, if nodes are users, items might represent the websites they have visited
.)\commentFT{Introducing the notion of items, as \HLSH is designed for item-based datasets.}
Despite being one of the simplest models in data analysis, computing an exact KNN graph remains extremely costly, incurring,
for instance, a quadratic number of similarity computations under a brute-force strategy.

Many applications, however, only require a reasonable approximation of a KNN graph, as long as this approximation can be produced rapidly. This is, for instance, true of online news recommenders, in which the use of fresh data is of utmost importance, or of machine learning techniques that use the KNN graph as a first preliminary step~\cite{chen2009fast,Tremblay:SignalProcessing2016}.
Existing approximate KNN graph algorithms essentially fall into two families, that each uses different strategies to drastically reduce the number of similarities they compute: \emph{greedy incremental solutions}~\cite{KIFF,Hyrec,NNDescent,bratic2018nn}, and \emph{partition-based techniques} (that include the popular LSH algorithm~\cite{LSH1,LSH2}).\commentFT{Mentioning LSH as early as possible, so that readers do not get the impression we are avoiding the comparison.}

Greedy incremental solutions are currently among the best performing KNN graph construction algorithms, and exploit a local incremental search: they start from an initial random $k$-degree graph, which they greedily improve by traversing neighbors-of-neighbors links.
Although they generally perform best\commentFT{toning down the `best' argument, to deflect criticism that these solutions are old etc.}, greedy approaches are critically hampered by their initial random start:
similar nodes that lie close to one another in the \emph{final KNN graph} are connected in the \emph{initial random graph} to dissimilar nodes. 
In other words, these greedy approaches present a poor \emph{initial graph locality}: 
nodes that are similar tend to be separated by long paths in the initial graph.\commentFT{Redundant with what precedes.}
As a result, greedy algorithms must initially compute many similarities between 
unrelated nodes, which adds a costly overhead for little to no benefit.

Partition-based algorithms~\cite{LSH1,LSH2,chen2009fast} avoid this problem by clustering nodes\commentFT{rather than users} before solving locally the problem, under a classic divide-and-conquer strategy.
Unfortunately, a clustering that is both good and efficient is difficult to achieve. 
\LSH~\cite{LSH1,LSH2} for instance relies on 
hash functions that tend to fragment sparse datasets
with large dimensions (from $\sim 10^3$ to $\sim 10^5$ in our experiments), which are typical of many on-line applications manipulating users and items.\commentOR{dimension instead of item sets?}\commentFT{We haven't mentioned items yet, so this comes as a surprise. No sure how to solve this.} Traditional clustering techniques such as k-means~\cite{macqueen1967some} are similarly ill-fitted, as they require many similarity computations, which are precisely what we seek to avoid.

In this paper, we reconcile both perspectives with {\bf \hyperLSH} (\HLSH for short), a KNN graph algorithm for item-based datasets\commentFT{stressing that we need items, so that it does not come as a surprise in sec 2} that boosts its initial graph locality by exploiting a novel, fast and accurate clustering scheme, dubbed {\bf\FMH{}}. \FMH{} 
does not require any similarity computations (similarly to \LSH) while avoiding fragmentation 
(similarly to k-means).
\FMH{} leverages {\bf fast random hash functions} and employs a {\bf new recursive splitting mechanism} to balance clusters, for optimal parallelism, and minimal synchronization between the involved threads in a parallel implementation.
\commentFT{Trying to inject threads early, as they come up later on.}

We present an extensive evaluation of \HyperLSH, performed on six real datasets, which confirms that our proposal significantly outperforms existing approaches, including LSH, yielding speed-ups ranging from $\times1.12$ (against \LSH on \MLone) to $\times 4.42$ (against \Hyrec, a state of the art greedy KNN algorithm~\cite{Hyrec}, on \AM) while incurring only a negligible loss in terms of KNN quality.


In the following, we first present the context of our work and our approach (Sec.~\ref{chap:HyperLSH:approach}). We then formally analyze the novel clustering scheme at the core of our proposal (Sec.~\ref{subsec:HyperLSH:property}). We present our evaluation procedure (Sec.~\ref{chap:HyperLSH:setup}) and our experimental results (Sec.~\ref{chap:HyperLSH:setup}); before reporting on factors impacting our solution (Sec.~\ref{sec:HyperLSH:sensitivity}). We finally discuss related work (Sec.~\ref{chap:HyperLSH:related}) and conclude (Sec.~\ref{chap:HyperLSH:conclusions}).

\section{\hyperLSH}
\label{chap:HyperLSH:approach}

For ease of exposition, we consider in the following 
that nodes are {\it users} associated with {\it items} (e.g. web pages, movies, locations).

\subsection{Notations and problem definition}
\label{sec:problem-definition}


We note $U = \{u_1,...,u_n\}$ the set of all users, and $I = \{i_1,...,i_m\}$ the set of all items. The subset of items associated with user $u$ (a.k.a. her {\it profile}) is noted $\itemsOfU{u} \subseteq I$. $\itemsOfU{u}$ 
is generally much smaller than $I$ (the universe of all items).

Our objective is to approximate a k-nearest-neighbor (KNN) graph over $U$ (noted $G_{\operatorname{KNN}}$) according to some similarity function $sim \in \mathbb{R}^{U \times U}$ computed over user profiles:
$$sim(u,v) =  \profileSim(\itemsOfU{u},\itemsOfU{v})$$

\noindent where $\profileSim$ may be any similarity function over \dictionariesOrSets that is positively correlated with the number of common items between the two \dictionariesOrSets, and negatively correlated with the total number of items present in both \dictionariesOrSets. These requirements cover some of the most commonly used similarity functions in KNN graph construction applications, such as cosine or the Jaccard similarity. We use the Jaccard similarity in the rest of the paper~\cite{Rijsbergen:1979:IR:539927}: $$sim(u,v)=J(P_u,P_v) = \frac{|P_u\cap P_v|}{|P_u \cup P_v|}$$

A KNN graph $\gknn$ connects each user $u\in U$ with a set $\knnFunc{u}$ (the `KNN' of $u$ for short) which contains the k most similar users to $u$, with respect to the similarity function $sim(u,-)$.

Computing an exact KNN graph is particularly expensive: a brute-force exhaustive search 
requires $O(|U|^2)$ similarity computations. 
Many scalable approaches, therefore, seek to construct {\it an approximate KNN graph}  $\agknn$, i.e., to find for each user $u$ a neighborhood $\aknn{u}$ that is as close as possible to an exact KNN neighborhood~\cite{Hyrec,NNDescent,bratic2018nn}.
The meaning of `close' depends on the context, but in most applications, a good approximate neighborhood $\aknn{u}$ is one whose aggregate similarity (its {\it quality}) comes close to that of an exact KNN set $\knnFunc{u}$.



We capture how well the average similarity of an approximated graph $\agknn$ compares against that of an exact KNN graph $\gknn$ with the {\it average similarity} of $\agknn$,
defined as the average similarity observed on the graph's edges: 
%
\begin{equation}
  \qualityFunc(\agknn) =  \frac{1}{k \times n}\sum \limits_{\substack{(u,v)\in U^2:\\\; v \in  \aknn{u}}} \profileSim(\itemsOfU{u},\itemsOfU{v}),
\end{equation}
%
We then define the {\it quality} of $\agknn$ as the ratio between its average similarity and the average similarity of an exact KNN graph $\gknn$:
\begin{equation}
  \qualityRatioFunc(\agknn) =  \frac{\qualityFunc(\agknn)}{\qualityFunc(\gknn)}.
\end{equation}
%
A quality close to 1 indicates that the approximate KNN graph 
can replace the exact one with little impact in most applications. 

With the above notations, we can summarize our problem as follows:  for a given dataset $(U,I, (\itemsOfU{u})_{u\in U})$ and item-based similarity $\profileSim$, we wish to compute an approximate 
$\agknn$ in the shortest time with the highest overall quality.

\subsection{Intuition}

{\bf Poor initial graph locality is a significant problem.}
\commentFT{Trying to tighten up the argument, so that the intuition comes faster.}
Some of today's best performing approaches for KNN graphs use a greedy strategy~\cite{Hyrec,NNDescent,bratic2018nn}: starting from a random initial $k$-degree graph, those algorithms incrementally seek to improve each user's neighborhood by exploring neighbors-of-neighbors links. Unfortunately, this random start tends to separate similar nodes by long paths in the initial graph. This disconnect between \emph{similarity} (which captures `closeness' from the application's point of view), and \emph{initial graph-distance} (in terms of shortest path between two nodes), means greedy approaches initially suffer from a poor \emph{graph locality}. This phenomenon is particularly marked in the first few iterations of greedy algorithms, in which neighbors-of-neighbors tend to be random, 
leading to spurious similarity computations, and a slow convergence~\cite{KIFF}.

Beyond convergence, a poor initial graph also hampers the concrete execution of greedy KNN graph approaches,
for practical reasons linked to the memory management of modern computers. Because their initial graph is random, greedy approaches must iterate through users in an arbitrary order, usually employing parallelization (a.k.a. multithreading) on a multicore machine (as is typical for KNN graph computations). This arbitrary order prevents greedy approaches from fully exploiting the cache hierarchies of modern hardware: each thread accesses unrelated users, with unrelated profiles and neighborhoods, and cannot benefit from memory locality (the tendency of programs to access the same memory areas over short time windows). This low memory locality lowers the hit rates of processor caches and reduces performance.

\medskip

\noindent{\bf \HyperLSH substantially improves the graph locality of its initial configuration, 
using an approximate, yet cheap and fast procedure}. 
This basic principle 
is illustrated in Figure~\ref{fig:HyperLSH:locality}.
Whereas standard greedy approaches (left) initially connect each user (in blue) to $k$ random neighbors (in red), \HyperLSH partitions users into small sub-datasets (also called clusters in the following), in which similar neighbors can be selected (Fig.~\ref{fig:HyperLSH:localityHyperLSH}), 
leading to much faster computation times.
We then assign each cluster to a dedicated thread, that computes its (partial) KNN graph 
in isolation, improving parallelism.
Merging all resulting partial graphs produces the final global KNN graph. 

\begin{figure}[tb]
  \captionsetup[subfloat]{farskip=0pt}  \center
  \subfloat[Traditional greedy approaches]{
    \includegraphics[width=0.43\linewidth]{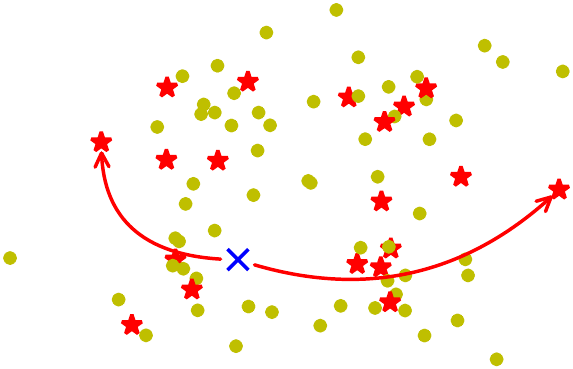}
    \label{fig:HyperLSH:localityHyrec}}
  \hfil
  \subfloat[\hyperLSH]{
    \includegraphics[width=0.43\linewidth]{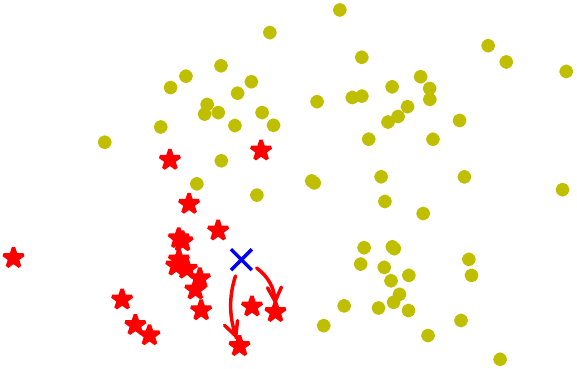}
  \label{fig:HyperLSH:localityHyperLSH}}
  \caption{Graph locality on a toy dataset (shown in 2D). 
    A given user (in blue) starts with unrelated neighbors (in red) with traditional approaches (a). \hyperLSH ensures a much higher initial graph locality (b).
}
  \label{fig:HyperLSH:locality}
\end{figure}

\medskip

\noindent {\bf Clustering is key to both performance and quality.} 
A core challenge when applying the above strategy consists in \textit{(i)} grouping together similar users to produce good sub-datasets, while \textit{(ii)} doing so on a tight computational budget.
This is hard, as most clustering techniques for item-based datasets either tend to fragment users in a large number of buckets (e.g. \LSH~\cite{LSH1,LSH2}),
or incur many similarity computations (e.g. k-means~\cite{macqueen1967some}).\commentFT{Ideally we would like to refer back to partition-based techniques, as discussed in the intro, e.g. something along the lines, `Taking inspiration from partition-based techniques, ...', or `The idea of partitioning users to produce KNN graphs faster~\cite{LSH1,LSH2,chen2009fast} is not new. \hyperLSH (what is different)', or `\hyperLSH applies partitioning (at the core of other KNN approaches such as LSH) to greedy approaches for KNN graph construction. Its core novelty lies in its partitioning techniques, that is both extremely fast and balanced.' At the same time, we don't want to draw too much attention to LSH. Not sure how to do this.}

\medskip

\noindent{\bf Fast but approximate clustering does the job.}
To overcome this challenge, we introduce \emph{\FMH{}}, a novel, fast-to-compute hashing scheme that we use at the core of \hyperLSH to group users into highly-local and balanced sub-datasets.
\FMH does not require any similarity computation between users, yet remains extremely lightweight to compute. \FMH exploits two ideas: \emph{(i)} \emph{redundant random hashing} on items for speed and graph locality, 
and \emph{(ii)} \emph{recursive splitting} for load balancing and parallelism.%
\commentFT{Move the following to later, when we present details.}
\commentFT{The use of random hash functions causes however collisions: users with no item in common may be hashed into the same cluster, resulting in superfluous similarities.
This is the price to pay for having a fast-to-compute, but approximate, hashing scheme.
}

\medskip

\noindent{\bf Introducing redundancy to compensate for approximate clustering.}
Because we use random hash functions, similar nodes may still end up within different clusters, preventing them from becoming neighbors in the final global KNN graph, and resulting in a poor KNN approximation. 
To mitigate this risk, we use multiple hash functions, thus increasing the probability that similar neighbors end up within the same cluster at least once.

\medskip

\noindent{\bf Recursively splitting large clusters to increase parallelism.}
By default, 
some clusters might become quite large, and slow down the whole computation.
To avoid this problem, we introduce a recursive load balancing mechanism that exploits the same \FMH scheme and repeatedly splits clusters that are larger than a given threshold $\nbHashTry$. 
\commentFT{It might be good to announce what comes afterward.}

\subsection{\HyperLSH: Overview}
\commentFT{I've moved this overview subsection up, to provide a plan for the rest of the section, which is quite long with many subsections. We can move it back to the end if you feel this does not work.}
Building on the previous intuitions, \HyperLSH works in three steps, which we present in detail in the remaining of this section:
\begin{itemize}
\item \underline{\bf Step $\bf 1$}: \textbf{Clustering}.
  The dataset is clustered in $\nbHash \times b$ clusters, using \FMH functions, where $t$ is the number of hash functions, and $b$ the number of clusters per hash function. (We return later on to the effect of these two parameters on the algorithm's speed and quality.)
  Clusters whose size exceeds $\nbHashTry$ are recursively split.
\item \underline{\bf Step $\bf 2$}: \textbf{Scheduling and KNN graph computation}.
  The clusters are processed in parallel to produce a partial KNN for each of them. The parallel computation uses a greedy scheduling heuristic to balance work between computing cores.
\item \underline{\bf Step $\bf 3$}: \textbf{Merging}.
  The resulting partial KNN graphs are merged.
\end{itemize}
The resulting KNN graph is returned as the KNN graph of the whole dataset.

\subsection{Step 1: Clustering with \FMH}
\label{subsec:HyperLSH:hash}
The \FMH scheme first projects each item $i\in I$ onto a hash value $h(i)$ using a generative hash function $h : I \rightarrow \llbracket 1,\nbBuckets \rrbracket$. 
The hash $H(u)$ of a user $u$ is then taken as the minimum hash value among $u$'s items:
\begin{equation}
  H(u)=\min_{i\in P_u} h(i).
\end{equation}

\newcommand{\overBraceAdjust}{-0.3em} 
\newcommand{\clusterIdAdjust}{-9pt}
\newcommand{\interHashAdjust}{-0.75em}
\newlength{\txtwd}

\commentOR{Example removed}
\commentFT{If we need space to include the theoretical analysis, I think we can remove this example.}
For example, consider the following hash function $h$ over an item set $I=\{i_1,i_2,i_3,i_4,i_5\}$, with $\nbBuckets=3$. 
\begin{center}
  \begin{tabular}{l@{\hspace{0.2em}}c}
  \vertbrace{5}{\rotatebox[origin=c]{-90}{$h$}}{5em}{0.em} &
    $i_1\rightarrow 2$\\
  & $i_2\rightarrow 3$\\
  & $i_3\rightarrow 2$\\
  & $i_4\rightarrow 1$\\
  & $i_5\rightarrow 3$
  \end{tabular}
\end{center}
\noindent If we apply \FMH to two users $u$ and $v$ whose profiles are given as
\begin{center}
  $\itemsOfU{u} = \{i_1,i_2,i_3\},$\\
  $\itemsOfU{v} = \{i_3,i_4,i_5\},$
\end{center}
we obtain, using the associated \FMH $H$,
\begin{center}
  $H(u)=min\{h(i_1),h(i_2),h(i_3)\}=min\{2,3,2\}=2,$\\
  $H(v)=min\{h(i_3),h(i_4),h(i_5)\}=min\{2,1,3\}=1,$
\end{center}
\noindent yielding the clustering configuration shown in Figure~\ref{fig:HyperLSH:example_cluster}. 


\begin{figure}[h]
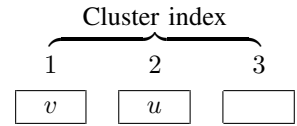

  \center
\begin{tabular}{|c| c |c| c |c|}
\multicolumn{5}{c}{$\overbrace{\hspace{8em}}^{\text{\normalsize{Cluster index}}}$}\\[\overBraceAdjust]
  \multicolumn{1}{c}{$1$} & \multicolumn{1}{c}{} & \multicolumn{1}{c}{$2$} & \multicolumn{1}{c}{} & \multicolumn{1}{c}{$3$}\\[\clusterIdAdjust]
\multicolumn{1}{c}{\hspace*{15pt}} & \multicolumn{1}{c}{} & \multicolumn{1}{c}{\hspace*{15pt}} & \multicolumn{1}{c}{} & \multicolumn{1}{c}{\hspace*{15pt}}\\
\hhline{*{2}{|-|~}-}
$v$ & & $u$ & & \\
  \hhline{*{2}{|-|~}-}

\end{tabular}
\caption{Clustering of $u$ and $v$ with $\nbBuckets=3$ clusters.}
\label{fig:HyperLSH:example_cluster}
\end{figure}
\commentFT{Removing the intermediate title `Clustering: \FMH in action', as this makes otherwise Section~\ref{subsec:HyperLSH:hash} extremely short.}
\label{subsec:HyperLSH:clustering}
\noindent $H(u)$ determines $u$'s cluster under $h$, resulting in $b$ clusters for each generative function, what we have termed a clustering configuration.
We use $\nbHash$ distinct generative functions, to produce $\nbHash$ clustering configurations and a total of $t\times\nbBuckets$ clusters.%
\commentFT{Detail, I think not needed at this stage. `In practice, we change the seed of the hash function used to hash the items to produce different \FMH functions.'}
The use of multiple hash functions reduces the risk that two similar nodes are never hashed into the same cluster, 
an event whose probability 
decreases exponentially with the number of hash functions $\nbHash$.


  

As an example, consider the earlier example of Section~\ref{subsec:HyperLSH:hash}.
We still have $I=\{i_1,i_2,i_3,i_4,i_5\}$ and $\nbBuckets=3$ and we are still interested in the two users $u$ and $v$.
We rename the hash function $h$ by $h_1$ and $H$ by $H_1$.
We consider another hash function $h_2$ and its associated \FMH $H_2$:

\begin{center}
  \begin{tabular}{l@{\hspace{0.2em}}c}
  \vertbrace{5}{\rotatebox[origin=c]{-90}{$h_2$}}{5em}{0.em} &
                                   $i_1\rightarrow 1$\\
  & $i_2\rightarrow 3$\\
  & $i_3\rightarrow 3$\\
  & $i_4\rightarrow 2$\\
  & $i_5\rightarrow 1$\\
  \end{tabular}
  \medskip

  $H_2(u)=min\{h_2(i_1)=1,h_2(i_2)=3,h_2(i_3)=3\}=1$\\
  $H_2(v)=min\{h_2(i_3)=3,h_2(i_4)=2,h_2(i_5)=1\}=1$
  \medskip
  
\hspace{-2em}
\begin{tabular}{@{}p{2em}@{}|c| c |c| c |c|}
\multicolumn{1}{c}{} &\multicolumn{5}{c}{$\overbrace{\hspace{8em}}^{\text{\normalsize{Cluster index}}}$}\\[\overBraceAdjust]
\multicolumn{1}{c}{} &  \multicolumn{1}{c}{$1$} & \multicolumn{1}{c}{} & \multicolumn{1}{c}{$2$} & \multicolumn{1}{c}{} & \multicolumn{1}{c}{$3$}\\[\clusterIdAdjust]
\multicolumn{2}{c}{\hspace*{15pt}} & \multicolumn{1}{c}{} & \multicolumn{1}{c}{\hspace*{15pt}} & \multicolumn{1}{c}{} & \multicolumn{1}{c}{\hspace*{15pt}}\\
  
  \hhline{~*{2}{|-|~}-}
$H_1$   &  $v$ & & $u$ & & \\
  \hhline{~*{2}{|-|~}-}
\multicolumn{6}{c}{}\\[\interHashAdjust]
  \hhline{~*{2}{|-|~}-}
$H_2$ & $u,v$ & & & & \\
  \hhline{~*{2}{|-|~}-}
\end{tabular}
\end{center}
\medskip


\noindent Because $H_1(u)=2 \neq H_1(v)=1$,\commentOR{We can remove the figure if more space is required} users $u$ and $v$ are mapped into different clusters in the first hashing configuration defined by $H_1$, but as $H_2(u)=H_2(v)=1$, they appear in the same cluster in the second configuration corresponding to $H_2$. The fact that $u$ and $v$ share one item ($i_3$), means they have a non-zero probability of appearing in the same cluster, even though we use random hash functions. (We characterize this property more precisely in Section~\ref{subsec:HyperLSH:property}.)


\newcommand{\cluster}{C}
\newcommand{\clusterIndex}{\eta_\cluster}
\newcommand{\clusterSize}{|\cluster|}
\newcommand{\clusterSet}{\mathcal{\cluster}_i}
\newcommand{\clusterSetProc}{\mathcal{\cluster}_p}

\SetKwComment{Comment}{$\rhd$}{}
\newcommand{\myCommandSty}[1]{\small\it #1}
\SetCommentSty{myCommandSty}

\begin{algorithm}[tb]
    \DontPrintSemicolon
    \For(\Comment*[f]{$t$ hash functions}){$i \in \llbracket 1,\nbHash \rrbracket$}{
    $B_i \leftarrow \text{new }Set \langle U \rangle [\nbBuckets]()$\;
    }
    \For(\Comment*[f]{clustering all users}){$u \in U$  {\bf and} $i \in  \llbracket 1,t \rrbracket$}{
    $B_i[H_i(u)] \leftarrow B_i[H_i(u)] \cup \{u\}$\Comment*[r]{\FMH}
    }
    \Return $(B_i)_{i\in \llbracket 1,\nbHash \rrbracket}$\;
  \caption{Step $1$ of \HLSH: the clustering}
  \label{algo:hyperlsh:clustering}
\end{algorithm}

Algorithm~\ref{algo:hyperlsh:clustering} shows the pseudocode of the clustering mechanism used by \HyperLSH. 
The variables $(\cluster_i)_{i\in \llbracket 1,\nbHash \rrbracket }$ are arrays of clusters.
There are $\nbHash$ of them, one for each hash function.
Each $\cluster_i$ is of size $\nbBuckets$, the number of cluster per hash function $H_i$.
Each $\cluster_i[j]$ is a cluster, represented by a set of users, so that $\bigcup\limits_{j=1}^{\nbBuckets} \cluster_i[j] = U$.


\medskip

\noindent {\bf Balancing large clusters through recursive splitting.
}
Using a minimum in the \FMH function, unfortunately, introduces a bias towards the clusters of low indices.
\commentFT{Removing the following example, as link to bias is not directly clear.}
\commentFT{`For instance, consider two users $u$ and $v$ such as $P_u=\{i_1,i_2,i_3\}$ and $P_v=\{i_1,i_2,i_3,i_4\}$ and a hash function $h$ such as $\forall j \ne 4,h(i_j)\geq 2$ and $h(i_4)=1$.
These users would end up in different clusters, although the presence of $i_4$ is the only difference in their profiles.
This is because $i_4$ is being hashed to a lower number than the other elements.'}
%
This bias is pervasive but particularly marked if highly popular items are hashed into one of the first clusters. In such a case, this cluster is likely to end up being much larger than the others in the same clustering configuration, many of which will be empty.

Highly unbalanced clusters tend to self-defeat the very purpose of the clustering step. This is because computing the partial KNN graph of a very large cluster can be almost as costly as that of the whole dataset. When this happens, the whole parallel computation is delayed, limiting the benefits of multi-threading.
To avoid this situation, we recursively split overlarge clusters, by using the fact that \FMH functions are extracted from hash values initially computed on items. 
More specifically, if a cluster $\cluster$ with index $\clusterIndex$ is larger than a threshold parameter $\nbHashTry$, we compute a second \FMH value $H\backslash \clusterIndex(u)$ for each of its users $u\in \cluster$, by ignoring $\cluster$'s index $\clusterIndex$ (i.e. the hash value that produced $\cluster$) when computing the hash values of $u$'s items:
$$ H\backslash \clusterIndex(u)=\min_{i\in P_u,\, h(i)>\clusterIndex} h(i) $$
\noindent where $h$ is the generative hash function that underlies the clustering configuration containing $\cluster$. The users of $\cluster$ are then distributed among new clusters, one for each new hash value, with two exceptions. Users who have a single item (for whom $H\backslash \clusterIndex$ is undefined) and users who are alone in their new cluster remain in $\cluster$.
The resulting clusters are again recursively split if their size exceeds $\nbHashTry$.
In summary, we split the users of a large cluster according to a second item, producing smaller and more refined clusters.

\begin{figure}[tb]
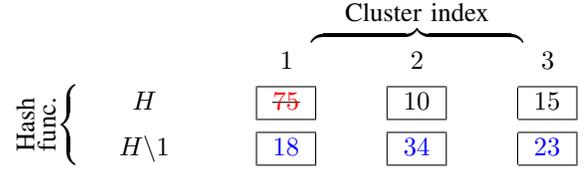

  \center
  \begin{tabular}{c c c |c| c |c| c |c|}
    
    & \multicolumn{1}{c}{} & \multicolumn{1}{c}{} & \multicolumn{5}{c}{$\overbrace{\hspace{8em}}^{\text{\normalsize{Cluster index}}}$}\\[\overBraceAdjust]
    &  & \multicolumn{1}{c}{}  & \multicolumn{1}{c}{$1$} & \multicolumn{1}{c}{} & \multicolumn{1}{c}{$2$} & \multicolumn{1}{c}{} & \multicolumn{1}{c}{$3$}\\[\clusterIdAdjust]
\multicolumn{1}{c}{} & \multicolumn{1}{c}{} & \multicolumn{1}{c}{\hspace*{15pt}} & \multicolumn{1}{c}{} & \multicolumn{1}{c}{\hspace*{15pt}} & \multicolumn{1}{c}{} & \multicolumn{1}{c}{\hspace*{15pt}}\\
    
\hhline{~~~*{2}{|-|~}-}
\vertbrace{3}{\substack{\text{\normalsize Hash}\\\text{\normalsize func.}}}{3em}{0.6em}
& $H$ & &  \sout{\color{red}$75$}
  & & 
      $10$
  & &
      $15$
  \\
\hhline{~~~*{2}{|-|~}-}
\multicolumn{1}{c}{} & \multicolumn{1}{c}{} & \multicolumn{1}{c}{} & \multicolumn{1}{c}{} & \multicolumn{1}{c}{}\\[\interHashAdjust]
\hhline{~~~*{2}{|-|~}-}
& $H\backslash 1$ & &
  {\color{blue}$18$}
& & 
  {\color{blue}$34$}
& & 
  {\color{blue}$23$}
\\
\hhline{~~~*{2}{|-|~}-}
\end{tabular}
\caption{Recursive splitting of clusters ($\nbBuckets=3$). In the initial clustering (first line), obtained with $H$, 
the first cluster $\cluster_H[1]$ exceeds the threshold of $\nbHashTry=40$ users, and is therefore split, by applying $H\backslash 1$, 
which only keeps item hashes higher than $1$. Users with no items being hashed to $2$ or $3$ remain in the first cluster. 
}
\label{figure:HyperLSH:2clusterings}
\end{figure}

Figure~\ref{figure:HyperLSH:2clusterings} illustrates this recursive splitting strategy on a cluster configuration of $|U|=100$ users, $\nbBuckets=3$ and $\nbHashTry=40$.
The first line represents the initial clustering, obtained with $H$, {\it before} the splitting.
The boxes represent the clusters and the number in each box the size of the cluster.
This initial clustering is highly unbalanced: the first cluster contains most of the users ($75$) while the others are nearly empty.
Since its size is higher than $\nbHashTry=40$, the first cluster is split into new clusters, shown on the second line.
The clusters of the second line are obtained using $H \backslash 1$, the \FMH $H$ keeping only hashes higher than $1$, on the $75$ users.
The first cluster is composed of users with no items being hashed to $2$ or $3$.
As none of the new clusters contains more than $40$ users, the splitting stops.
With the new clusters (shown in the second line), the new clustering configuration is more balanced, at the cost a few more clusters ($5$ instead of $3$).


The lower the threshold size $\nbHashTry$, the more balanced the final $\nbHash$ cluster configurations, thus the faster the computation.
Still, small clusters increase the chance of similar users never appearing in the same cluster, potentially hurting the quality of the final global KNN graph.
In practice, we choose $\nbHashTry=2000$.

\subsection{Comparison with LSH/MinHash}
\label{subsection:comparison_LSH}
Although \FMH can be understood as a randomized variant of the popular MinHash 
algorithm~\cite{broder1997resemblance,li_theory_2011}, often used with LSH~\cite{LSH1,LSH2}, key differences in its design lead to starkly different properties.
{\bf The use of random hash functions on a bounded discrete interval} $\llbracket 1,\nbBuckets \rrbracket$  considerably reduces the resulting number of buckets (we use $b=4096$ by default in our experiments), which is otherwise determined by the size  $|I|$ of the items universe with MinHash (which can be as high as 203,030 in the datasets we consider). Fewer buckets limit the dispersion of users in many small sub-datasets and increase the chances of finding good KNN neighbors while aligning better with the needs of a parallel implementation.
Our choice of a small hashing space, however, tends to cause collisions and to produce 
unbalanced 
clusters.

{\bf Recursive splitting} not only mitigates this second problem 
but also caps the maximum size of individual buckets, making the local KNN graph computations faster.
Recursive splitting is also tightly linked to
the small size of our hashing space.
While it could in principle also be applied to MinHash,
it would further fragment users into a still larger number of buckets, increasing the problem of dispersion mentioned earlier for the kind of sparse and large-dimensional datasets we consider.%
\commentFT{R}

\subsection{Step 2: Scheduling and local 
  computation}
\label{subsec:HyperLSH:step2}
\commentFT{Merging this subsection and the next, to avoid subsections that are too short.}
Recursively splitting clusters reduces gross discrepancies between cluster sizes, but the final clusters might still remain unbalanced.
To prevent an unbalanced workload among the threads of a parallel architecture, we apply some light-weight work scheduling.
The clusters are stored in a synchronized, decreasing priority queue, ordered according to their size.
We then use a basic thread pool to computes the KNN graph of each cluster in the queue, starting with the largest clusters and working down
the priority queue until it becomes empty.

{\bf The partial KNN graph of each  cluster $\cluster$ can be computed using any approach and does not need to be synchronized with any other computation}.
In our prototype, we use a hybrid solution
that switches between a brute force approach and a greedy KNN graph algorithm
depending on the number of users $\clusterSize$ in the cluster. 
(In practice we use \Hyrec~\cite{Hyrec}, see Sec.~\ref{setup:algos}, but any other KNN graph algorithm can be used.)
To determine a threshold value for the switch, we estimate the expected number of similarity computations for each approach and pick the least expensive.
The brute force approach computes $\frac{\clusterSize \times (\clusterSize-1)}{2}$ similarities, while \Hyrec's number of similarities is bounded by $\frac{\rho \times k^2 \times \clusterSize}{2}$, where $\rho$ is the number of iterations.
As a result, if $\clusterSize < \rho \times k^2$ we choose the brute force approach, \Hyrec otherwise.
Algorithm~\ref{algo:hyperlsh:computing} shows the pseudocode of the local KNN graph computation used by \HyperLSH.
In practice, we take $\rho=5$.
To further speed-up the computation, we use optimized versions of these algorithms that leverage a compact data structure~\cite{GFWWW} to provides a {\it fast-to-compute} estimation of Jaccard similarity values. 
In practice, this data structure, dubbed {\it \GF}~\cite{GFWWW}, summarizes each user's profile into a 64- to 8096-bit vector, which is then used to estimate Jaccard similarity values.

\begin{algorithm}[tb]
    \DontPrintSemicolon
    \lIf{$\clusterSize < \rho \times k^2$}{
    \Return $BruteForce(\cluster)$}
    \lElse{
    \Return $Hyrec(\cluster)$}
  \caption{Step 3 of \HLSH: local KNN on a cluster \cluster}
  \label{algo:hyperlsh:computing}
\end{algorithm}

\subsection{Step 3: Merging the KNN graphs}
\label{subsec:HyperLSH:step3}
Once the cluster phase is over, we merge the partial KNN graphs obtained for each cluster, one by one, into a unique KNN graph $knn$.
Merging (Algorithm~\ref{algo:hyperlsh:conquer}) is performed at the granularity of individual users. 
Each user appears in $\nbHash$ different clusters and is connected to up to $\nbHash \times k$ neighbors.
For each cluster $\cluster$, and each user $u$ of $\cluster$, the KNN neighborhood $knn'(u)$ of $u$ in $C$'s partial KNN graph is added to $u$'s final neighborhood $knn(u)$, while only keeping the $k$ best neighbors so far in $knn(u)$. While doing so we are careful to reuse similarity values, to avoid redundant computations.
The resulting KNN graph $knn$ is returned.

\begin{algorithm}[tb]
    \DontPrintSemicolon
    $knn \leftarrow \text{new }  knn()$\;
    \For(\Comment*[f]{for each hash function}){$i \in \llbracket 1,\nbHash \rrbracket $}{
    \For(\Comment*[f]{and for each cluster}){$\cluster \in \clusterSet$}{
      $knn' \leftarrow \cluster.knn()$\Comment*[r]{compute the sub KNN for $\cluster$}
    \For(\Comment*[f]{merge this KNN into $knn(u)$}){$u \in knn'$} {
      \For{$(v,s) \in knn'(u)$}{
        $knn(u).add(v,s)$\;
        \Comment*[r]{$knn(u)$ is a heap bounded to size $k$}
    }
    }
    }
    }
    \Return $knn$\;
  \caption{Step $4$ of \HLSH: KNN merging}
  \label{algo:hyperlsh:conquer}
\end{algorithm}

\section{Theoretical Properties}
\label{subsec:HyperLSH:property}

\newcommand{\JaccardShort}{J_{1,2}}
\newcommand{\JaccardLong}{J(P_1,P_2)}


The final neighbors of a user are selected from within the clusters in which this user appears. 
The quality of the final KNN graph is thus highly dependent on the ability of the hashing scheme to group similar users together.
In the following, we prove that the probability of two users being hashed into the same bucket, before any splitting occurs, is proportional to their Jaccard similarity up to some small error introduced by collisions.
\begin{theorem}\label{theo:probaSameFMH}
  The probability $\probNaked[H(u_1)=H(u_2)]$ that two users $u_1$, $u_2$ obtain the same \FMH value $H(\cdot)$ is lower bounded by the following inequality
  \begin{align}\label{eq:lowerBound}
  \JaccardShort - \frac{\nbColl}{\Cu} \leq \probNaked[H(u_1)=H(u_2)],
\end{align}
  where $\JaccardShort=\JaccardLong$ is the Jaccard similarity between $u_1$'s and $u_2$'s profiles, $P_1$ and $P_2$; $\Cu=|P_1\cup P_2|$ is the joint size of the two profiles; $\nbColl=\Cu-|h(P_1\cup P_2)|$ is the overall number of collisions occurring when projecting the two profiles onto $\llbracket 1,\nbBuckets \rrbracket$, and $h(\cdot)$ is the generative hash function underpinning $H(\cdot)$.
  If we further assume $\nbColl\leq\Cu/2$, then we have
  \begin{align}
  \probNaked[H(u_1)=H(u_2)]
  \leq \JaccardShort + 3\frac{\nbColl}{\Cu} + O\left(\frac{\nbColl}{\Cu}\right)^2.
\end{align}
\end{theorem}

\begin{proof}
  $\probNaked[H(u_1)=H(u_2)]$ is proportional to $|h(P_1)\cap h(P_2)|$, where $h(X)$ is the image of the set $X$ by $h$.
  This is because $H(u_1)=H(u_2)$ iff the minimum element of $h(P_1\cup P_2)$ happens to belong also to $h(P_1)\cap h(P_2)$. As the generative hash functions that send $P_1\cup P_2$ onto $h(P_1\cup P_2)$ as $h$ does are each equally probable over the space of generative hash functions, the probability that the minimum element of $h(P_1\cup P_2)$ belongs to $h(P_1)\cap h(P_2)$ is given by the ratio of the two sets' sizes.
  More precisely:
\begin{equation}
  \probNaked[H(u_1)=H(u_2)]=\frac{|h(P_1)\cap h(P_2)|}{|h(P_1\cup P_2)|}.
\end{equation}
For brevity, let us note $\Pin=P_1\cap P_2$ the set of items that are present in both profiles. Compared to $\Pin$,
collisions may both increase $h(P_1)\cap h(P_2)$ (if they occur between elements of $P_1\Delta P_2$, the symmetric difference of the users' profiles, i.e. the items that only appear in one of the two profiles), or decrease it (if they occur between elements of $\Pin$), yielding
\begin{equation}
  |\Pin| - \nbColl \leq|h(P_1)\cap h(P_2)|\leq |\Pin| + \nbColl,
\end{equation}
where $\nbColl$ is the number of collisions caused by $h$ on $P_1\cup P_2$.  Since by definition $|h(P_1\cup P_2)|=\Cu-\nbColl$, we get
\begin{align}
  \frac{|\Pin|/\Cu - \nbColl/\Cu}{1-\nbColl/\Cu} \leq
  &\probNaked[H(u_1)=H(u_2)] \leq 
  \frac{|\Pin|/\Cu + \nbColl/\Cu}{1-\nbColl/\Cu}\\
  \frac{\JaccardShort - \nbColl/\Cu}{1-\nbColl/\Cu} \leq
  &\probNaked[H(u_1)=H(u_2)] \leq 
\frac{\JaccardShort + \nbColl/\Cu}{1-\nbColl/\Cu}.
\end{align}
Since $\nbColl<\Cu$, we trivially have $\frac{\JaccardShort - \nbColl/\Cu}{1-\nbColl/\Cu}\geq \JaccardShort - \frac{\nbColl}{\Cu}$, thus proving~(\ref{eq:lowerBound}). If we assume $\nbColl\leq\Cu/2$, we can use the fact that $\frac{1}{1-x}\leq 1+2x$ when $x\in[0,\frac{1}{2}]$, which yields, together with $\JaccardShort\leq1$
\begin{align*}
  \frac{\JaccardShort + \nbColl/\Cu}{1-\nbColl/\Cu} \leq&
  \left(\JaccardShort + \frac{\nbColl}{\Cu}\right)
  \left(1+2\frac{\nbColl}{\Cu}\right)
  \leq
  \JaccardShort + 3\frac{\nbColl}{\Cu} + O\left(\frac{\nbColl}{\Cu}\right)^2
\end{align*}
\end{proof}

Theorem~\ref{theo:probaSameFMH} states that \FMH will tend to allocate the same hash to similar users, modulo some noise introduced by collisions.
In other words, {\bf the more similar two users are, the more likely they are to be allocated in the same clusters}.
We now bound the effect of collisions with the following concentration bound.

\newcommand{\averageProb}{\frac{\Cu-1}{2\nbBuckets}}
\begin{theorem}\label{theo:probaCollisionDensity}
  \commentFT{FT16Jun20: Replacing the old concentration bound with that of George, as provides results that are as tight, and is much simpler}
  The collision density $\nbColl/\Cu$ is upper bounded by the value $(1+d)\averageProb$ with a probability bounded by the following formula
  \begin{align}
    \begin{array}{l}\displaystyle
    \label{eq:ConcentrationGeorge}
    \probNaked\left[\frac{\nbColl}{\Cu} < (1+d)\averageProb\right] \geq
    %
    1 - \left(\frac{e^d}{(1+d)^{(1+d)}}\right)^{\textstyle\frac{\Cu(\Cu-1)}{2b}},
    \end{array}
  \end{align}
  where $d>0$ is a real positive value, and the other variables are defined as in Theorem~\ref{theo:probaSameFMH}.\label{theo:boundOnKappaEll}
  
\end{theorem}

\begin{proof}
  If we imagine we project each element of $P_1\cup P_2$ one after the other, we can define $\Cu$ 
random variables\commentFT{Removing independent, as I don't think the $X_k$ are independent.} \((X_k)_\Cu\) that are equal to \(1\) when the $k^\mathrm{th}$ element causes a collision, and to \(0\) otherwise. We have $\nbColl= \sum_{1 \leq k \leq \Cu} X_k$.
  By observing that $\prob{X_k=1} \leq \frac{k}{b}$, we can define \(\Cu\) independent random variables \((X'_k)_\Cu\), that are equal to \(1\) with probability \(\frac{k}{b}\), while upper bounding their corresponding \(X_k\) in all realizations. As a result, \(\nbColl\) is upper bounded by their sum, which we note \(S'\),
\begin{equation}\label{eq:nbColl}
\nbColl \leq \sum_{1 \leq k \leq \Cu} X'_k = S'
\end{equation}
We have $\expect{S'}=\sum_{k=0}^{\Cu-1}\frac{k}{b}=\frac{\Cu(\Cu-1)}{2b}$.
By applying the first Chernoff bound proposed by Mitzenmacher and Upfal~\cite{mitzenmacher2017probability} to $S'$, we obtain the following concentration bound for S': 
  \begin{align}
    \begin{array}{l}\displaystyle
    \label{eq:ConcentrationGeorge:onS}
    \probNaked\left[S' \geq (1+d)\frac{\Cu(\Cu-1)}{2b}\right] \leq
    %
    \left(\frac{e^d}{(1+d)^{(1+d)}}\right)^{\textstyle\frac{\Cu(\Cu-1)}{2b}},
    \end{array}
  \end{align}
where $d>0$ is a real positive value. Eq.~(\ref{eq:nbColl}) implies that $\probNaked\left[\nbColl \geq x\right] \leq \probNaked\left[S' \geq x\right]$ for any $x$. Applying this observation to (\ref{eq:ConcentrationGeorge:onS}) and taking the complement yields the theorem.
\end{proof}

\noindent As an example, if we apply Theorems~\ref{theo:probaSameFMH} and~\ref{theo:boundOnKappaEll} to the case of $\Cu=256$, $b=4096$ (some typical values of our experiments), and set $d=0.5$ 
, we obtain that 
$$
\JaccardShort - 0.078 \leq \probNaked[H(u_1)=H(u_2)] \leq \JaccardShort + 0.234
$$
with probability $0.998$ over the space of all hash functions $h$.
The left-hand side of the above equation impacts the quality of our approximation, by ensuring that pairs of similar users tend to be compared: such users show a high $\JaccardShort$ value, and have therefore a high probability to be hashed to the same bucket, and to be compared, since this probability is lower-bounded by a value which is close to their Jaccard similarity.
Conversely, the right-hand side controls the performance of our approach, by lowering the chances of comparing dissimilar users (and thus performing superfluous computations): the probability of such a comparison taking place is upper-bounded by $\JaccardShort$ (which is low for dissimilar users) plus a constant due to collisions that depends on $\Cu$ and $\nbBuckets$ (0.234 in our numerical example).


\section{Experimental Setup}
\label{chap:HyperLSH:setup}

\begin{table*}[tb]
  \caption{Description of the datasets used in our experiments}
  \label{tab:dataDescription}
  \small
  \advance\leftskip-4cm
  \centering
  \begin{tabular}{lrrrrrrc}
    \hline
    \bf Dataset & \bf Users & \bf Items & \bf Scale & \bf Ratings $> 3$ &  $|P_u|$  &  $|P_i|$  & \bf Density \\ \hline
    \MLone (\mlone)~\cite{movielens}  & 6,038 & 3,533 & 1-5 & 575,281 & 95.28 & 162.83 & 2.697\% \\
    \MLten (\mlten)~\cite{movielens}  & 69,816 & 10,472 & 0.5-5 & 5,885,448 & 84.30 & 562.02 & 0.805\% \\
    \MLtwenty (\mltwenty)~\cite{movielens}  & 138,362 & 22,884 & 0.5-5 & 12,195,566 & 88.14 & 532.93 & 0.385\% \\
    \AM (\am)~\cite{mcauley2013amateurs} & 57,430 & 171,356 & 1-5 & 3,263,050 & 56.82 & 19.04 & 0.033\% \\
    \dblp~\cite{DBLP} & 18,889 & 203,030 & 5 & 692,752 & 36.67 & 3.41 & 0.018\% \\
    \GW (\gw)~\cite{Gowalla} & 20,270 & 135,540 & 5 & 1,107,467 & 54.64 & 8.17 & 0.040\% \\
    
    \hline
  \end{tabular}
\end{table*}

\subsection{Datasets}
\label{sec:datasets}
We use six publicly available users/items datasets (Table~\ref{tab:dataDescription}) that cover a varied range of domains (movies and books reviews, co-authorship graphs, and geolocated social networks).

\subsubsection{Three MovieLens datasets} 
MovieLens~\cite{movielens} is a group of anonymous datasets containing movie ratings collected on-line between 1995 and 2015 by GroupLens Research \cite{GroupLens}.
The datasets 
contain movie ratings on a $0.5$-$5$ scale by users who have at least performed more than $20$ ratings. 
To compute the Jaccard similarity, we binarize these datasets by keeping only ratings that reflect a positive opinion (i.e. higher than $3$). 
We use three versions of the dataset, \MLone (\mlone), \MLten (\mlten) and \MLtwenty (\mltwenty), containing between $575,281$ and $12,195,566$ positive ratings. 

\subsubsection{The \AM dataset} 
\AM ~\cite{mcauley2013amateurs} (\am) is a dataset of movie reviews from Amazon collected between 1997 and 2012.
Ratings range from $1$ to $5$.
We restrict our study to users with at least $20$ ratings (positive and negative ratings) to avoid dealing with users with not enough data (this problem, called the {\it cold start problem}, is generally treated separately~\cite{lam2008addressing}).
After binarization, the resulting dataset contains $57,430$ users;  $171,356$ items; and $3,263,050$ ratings.

\subsubsection{DBLP}
\DBLP ~\cite{DBLP} is a dataset of co-authorship from the DBLP computer science bibliography.
In this dataset, both the user set and the item set are subsets of the author set.
If two authors have published at least one paper together, they are linked, which is expressed in our case by both of them
appearing in the profile of the other with a rating of `5'.
As with \am, we only consider users with at least $20$ ratings: the others are removed from the user set but not from the item set. 
The resulting dataset contains $18,889$ users, $203,030$ items, and $692,752$ ratings.

\subsubsection{Gowalla}
\GW ~\cite{Gowalla} (\gw) is a location-based social network.
As \DBLP, both user set and item set are subsets of the set of the users of the social network.
The undirected friendship link from $u$ to $v$ is represented by $u$ rating $v$ with a $5$.
As previously, only the users with at least $20$ ratings are considered. The resulting dataset contains $20,270$ users, $135,540$ items; and $1,107,467$ ratings.

\medskip

We use all datasets for the main performance evaluation (Sec.~\ref{subsec:HyperLSH:results}). We then focus on \MLten and \AM for the parameter sensitivity analysis (Sec.~\ref{sec:HyperLSH:sensitivity}). \MLten and \AM have a similar number of users ($69,816$ for \mlten, $57,430$ for \am) and ratings ($5,885,448$ for \mlten, $3,263,050$ for \am) but they differ by the size of their item set ($10,472$ for \mlten against $171,356$ for \am): \MLten is dense while \AM is sparse. This difference allows us to assess how sparsity 
impacts the performance and quality of \HyperLSH.

\subsection{Baseline algorithms and competitors}
\label{setup:algos}

We compare our approach against four competitors: a naive brute-force solution for reference, two state-of-the-art greedy KNN-graph algorithms (\NNDescent~\cite{NNDescent,bratic2018nn} and \Hyrec~\cite{Hyrec}), and \LSH~\cite{LSH1}. On each dataset, we use the fastest competitor as our main baseline (termed `baseline' or underlined in the following).

While many techniques exist to compute KNN graphs, the performance of most of them heavily depends on the properties of the dataset they are applied to.
For example, product quantization~\cite{jegou2010product} is designed to work well on datasets with a few hundred dimensions and dense values, i.e. in which each data point possesses a non-zero value in most dimensions. 
By contrast, our datasets are high-dimensional (ranging from $3,533$ to $203,030$ items), and sparse, containing only a few ratings per user, two characteristics that render them unsuitable for many existing KNN graph construction techniques. 
\LSH remains the state of the art for KNN graphs computation in high-dimensional spaces and is routinely used as a baseline in recent works~\cite{liu2019lsh,zheng2020pm,cai2019revisit}.
\NNDescent experimentally outperforms \LSH on high-dimensional datasets~\cite{NNDescent}.\commentFT{Removing comment about ``still being the subject of publications'' and instead adding \cite{bratic2018nn} each time we mention \NNDescent.} 
For a fair comparison, all competitors use the \GF compact datastructure~\cite{GFWWW} to compute similarity values (Section~\ref{subsec:HyperLSH:step2}).

\subsubsection{Brute force}
The \BruteForce competitor simply computes the similarities between every pair of profiles, performing a constant number of similarity computations equal to $\frac{n \times (n-1)}{2}$.
This algorithm suffers from a high complexity, but produces an exact KNN graph.

\subsubsection{Greedy approaches}
\label{subsec:greedy}
\emph{\NNDescent and \Hyrec}~\cite{Hyrec,NNDescent,bratic2018nn} are state-of-the-art greedy approaches that construct an approximate KNN graph by exploiting a local search strategy and by limiting the number of similarities computations.
They start from an initial random graph, which is then iteratively refined until convergence.
\NNDescent~\cite{NNDescent,bratic2018nn} and \Hyrec~\cite{Hyrec} mainly differ in their iteration procedure.
\NNDescent compares all pairs $(u_i,u_j)$ among the neighbors of $u$, and updates the neighborhoods of $u_i$ and $u_j$ accordingly.
By contrast, \Hyrec compares all the neighbors' neighbors of $u$ with $u$, rather than comparing $u$'s neighbors between themselves.
Both algorithms stop either when the number of updates during one iteration is below the value $\delta \times k \times |U|$, with a fixed $\delta$, or after a fixed number of iterations.

\subsubsection{\LSH}
Locality-Sensitive-Hashing (\LSH) \cite{LSH1} reduces the number of similarity computations by hashing each user into several buckets. The neighbors of a user $u$ are then selected among the users present in the same buckets as $u$.
To ensure that similar users tend to be hashed into the same buckets, \LSH uses min-wise independent permutations of the item set as its hash functions, similarly to the MinHash algorithm~\cite{broder1997resemblance}.
For fairness, we implement \LSH the same way as \HyperLSH: each hash function creates its own buckets, independently from each other, rather than having one bucket per item. This drastically decreases the number of similarity computations, resulting in a faster overall computation time while only inducing a small loss in quality.

\subsection{Parameter setup}
We compute KNN graphs with neighborhoods of size $k=30$, a standard value~\cite{NNDescent}.
By default, when using \hyperLSH, the number of clusters per hash functions $b$ is set to $4096$ and the number of hash functions $t$ to $8$, except for \dblp and \gw for which the number of hash functions is $15$.
The maximum size of clusters for the recursive splitting procedure is set to $N=2000$, except for \MLtwenty for which it is $N=4000$. 
Both maximum sizes for clusters are below the threshold that determines whether BruteForce or \Hyrec is used ($\rho \times k^2=4500$) in order to privilege \BruteForce which tends to deliver better sub-KNNs than \Hyrec.\commentFT{This sounds like we never use \Hyrec for sub-KNNs. Is it true?}
While locally computing the KNN graphs in clusters, we use 1024-bit \GF vectors (See Section~\ref{subsec:HyperLSH:step2}). 
Beyond these default values, we present a detailed sensitivity study of the effect of $b$, $t$, and $N$ on the performance of \hyperLSH in Section~\ref{subsec:HyperLSH:nbhashnbclusters}.

The parameter $\delta$ of \Hyrec and \NNDescent is set to $0.001$, and their maximum number of iterations to $30$.
The number of hash functions for \LSH is 10.

\subsection{Evaluation metrics}
We measure the performance of \hyperLSH and its competitors along two main metrics: {\it (i)} their computation {\it time}, and {\it (ii)} the {\it quality} ratio of the resulting KNN (Sec.~\ref{sec:problem-definition}).
As an example of application, we also use the KNN graphs produced by \hyperLSH to compute recommendations and compare the resulting \emph{recall} to recommendations obtained with an exact KNN graph computed with the brute force approach.
Throughout our experiments, we use a 5-fold cross-validation procedure and average our results over the 5 resulting runs.

\subsection{Implementation details and hardware}
We have implemented \BruteForce, \Hyrec, \NNDescent, \LSH, and \hyperLSH in Java 1.8. 
Our \FMH functions are computed using Jenkins’ hash function~\cite{jenkins1997hash}.
Our experiments run on a 64-bit Linux server with two Intel Xeon E5420@2.50GHz,
totaling 8 hardware threads, 32GB of memory, and a HHD of 750GB.
We use all 8 threads.
Our code is available online\footnote{https://gitlab.inria.fr/oruas/SamplingKNN}.

\section{Evaluation}
\label{sec:HyperLSH:eval}

We first discuss the raw performance of \hyperLSH, compared to the brute force approach, \LSH, \NNDescent, and \Hyrec (Sec. \ref{subsec:HyperLSH:results}).
We then evaluate the performance of the obtained KNN graphs when used to provide recommendations (Sec. \ref{subsec:HyperLSH:reco}).
We evaluate the impact of \FMH on \hyperLSH (Sec.~\ref{sec:HyperLSH:impact}) and finally assess the effect of the \GF data structure (Sec.~\ref{sec:impact-gf}).

\renewcommand{\myHrule}{\cline{3-7}}

\begin{table}[tb]
  \caption{Computation time and KNN quality. 
  Speed-ups are computed against the best baseline (underlined). \HyperLSH clearly outperforms all competitors, yielding speed-ups of up to $\times4.42$ 
  against the state of the art.
  } 
  \label{tab:time_KNNquality_HyperLSH}
  \renewcommand{\mycolsep}{\hspace{0.6em}}
  \centering
  \small
    \begin{tabular}{@{\hspace{-.5em}}r@{}cr@{\mycolsep}r@{\mycolsep}c@{\mycolsep}c@{\mycolsep}c}
      & & Algo & \bf Time (s) & \bf Gain ($\%$) & \bf Quality & \bf $\Delta$  \\
      \myHrule
      \vertbrace{23}{\mathit{datasets}}{28.5em}{0em}


&\dataset{\mlone}	

	    &	\Hyrec	&	$4.43$	&	-	&	$0.92$	&	- \\
&	&	\NNDescent	&	$10.98$	&	-	&	$0.93$	&	- \\
&	&	\underline{\LSH}	&	\underline{$2.96$}	&	-	&	\underline{$0.92$}	&	- \\
      &	&	\high \HLSH	(ours) &	\high $\bf 2.64$	&	\high $10.81$	&	\high $0.91$	&	\high $-0.01$ \\

      \datasetsep
&\dataset{\mlten}

	&	\underline{\Hyrec}	&	\underline{$109.98$}	&	-	&	\underline{$0.90$}	&	- \\
&	&	\NNDescent	&	$147.03$	&	-	&	$0.93$	&	- \\
&	&	\LSH	&	$255.33$	&	-	&	$0.94$	&	- \\
      &	&	\high \HLSH	&	\high $\bf 27.79$	&	\high $74.73$	&	\high $0.89$	&	\high $-0.01$ \\

      \datasetsep
&\dataset{\mltwenty}

	&	\underline{\Hyrec}	&	\underline{$289.23$}	&	-	&	\underline{$0.88$}	&	- \\
&	&	\NNDescent	&	$383.21$	&	-	&	$0.92$	&	- \\
&	&	\LSH	&	$1060.76$	&	-	&	$0.93$	&	- \\
      &	&	\high \HLSH	&	\high $\bf 106.25$	&	\high $63.26$ 	&	\high $0.89$	&	\high $+0.01$ \\

      \datasetsep
&\dataset{\am}

	&	\underline{\Hyrec}	&	\underline{$62.41$}	&	-	&	\underline{$0.93$}	&	- \\
&	&	\NNDescent	&	$91.24$	&	-	&	$0.95$	&	- \\
&	&	\LSH	&	$140.53$	&	-	&	$0.96$	&	- \\
      &	&	\high \HLSH	&	\high $\bf 14.11$	&	\high $77.39$	&	\high $0.95$	&	\high $+0.02$ \\

      \datasetsep
&\dataset{\dblp}

	    &	\Hyrec	&	$26.84$	&	-	&	$0.81$	&	- \\
&	&	\underline{\NNDescent}	&	\underline{$24.43$}	&	-	&	\underline{$0.82$}	&	- \\
&	&	\LSH	&	$37.80$	&	-	&	$0.86$	&	- \\
      &	&	\high \HLSH	&	\high $\bf 6.54$	&	\high $73.27$	&	\high $0.84$	&	\high $+0.02$ \\

      \datasetsep
&\dataset{\gw}

	&	\underline{\Hyrec}	&	\underline{$21.88$}	&	-	&	\underline{$0.78$}	&	- \\
&	&	\NNDescent	&	$26.05$	&	-	&	$0.79$	&	- \\
&	&	\LSH	&	$26.91$	&	-	&	$0.82$	&	- \\
&	&	\high \HLSH	&	\high $\bf 8.38$	&	\high $61.70$	&	\high $0.82$	&	\high $+0.04$ \\

      \myHrule
      
    \end{tabular}
\end{table}

\subsection{Computation time and KNN quality}
\label{subsec:HyperLSH:results}

The performances of \hyperLSH are summarized in Table~\ref{tab:time_KNNquality_HyperLSH} over the six datasets.
A part of the results is shown graphically in Figures~\ref{fig:HyperLSH:Time} and~\ref{fig:HyperLSH:KNNquality}.
In addition to those of \hyperLSH (noted \HLSH), the performances of 
\Hyrec, \NNDescent, and \LSH are also displayed for each dataset.
The best computation time is shown in bold, the time of the best baseline is underlined, and the speed-up is computed w.r.t. this best baseline.

\HyperLSH provides the best computation time on all the datasets.
\HyperLSH clearly outperforms all the approaches, providing speed-ups from $\times1.12$ ($-10.81\%$ on \mlone) to $\times4.42$ ($-77.39\%$ on \am) compared to the best baselines.
The KNN quality provided by \hyperLSH is similar to the one provided by the fastest approaches: it goes from a loss of $0.01$ (on \mlone) to a gain of $0.04$ (on \gw).


\begin{figure*}[tb]
  \captionsetup[subfloat]{farskip=0pt}  \center
  \subfloat[\MLtwenty]{
    \includegraphics[width=0.23\linewidth]{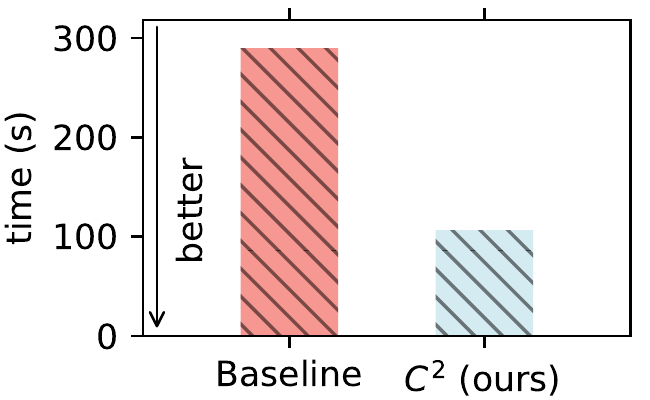}}\hfil
  \subfloat[\AM]{
    \includegraphics[width=0.23\linewidth]{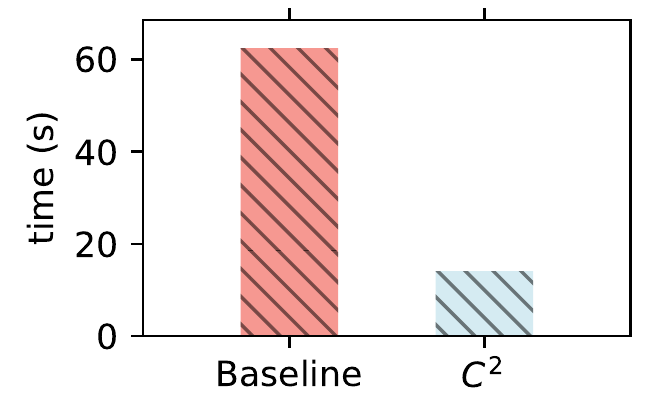}}\hfil
  \subfloat[\DBLP]{
    \includegraphics[width=0.23\linewidth]{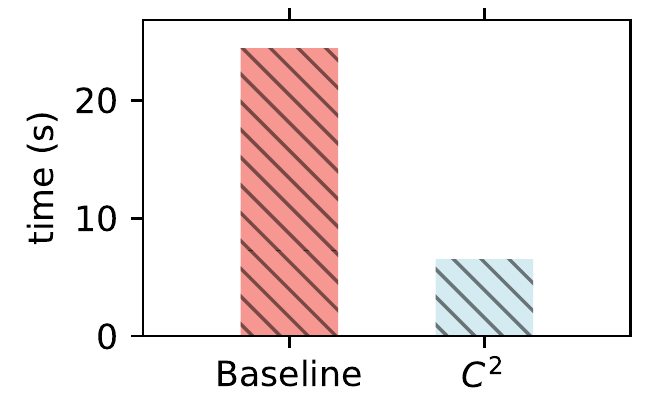}}\hfil
  \subfloat[\GW]{
    \includegraphics[width=0.23\linewidth]{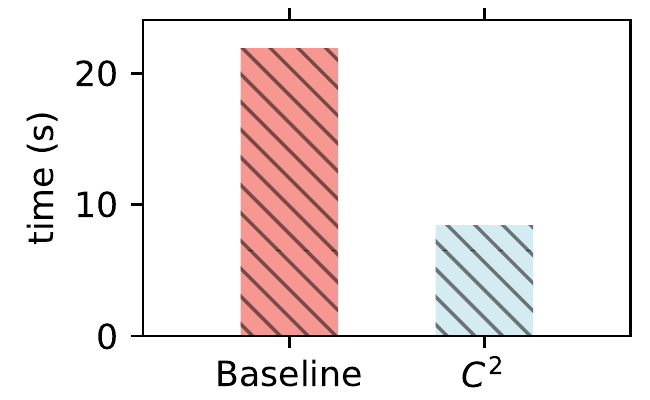}}
  \caption{Execution time of \hyperLSH (\HLSH) and the best competing approach (Baseline) for each dataset (lower is better). \HyperLSH clearly outperforms the best competitor across all datasets.}
  \label{fig:HyperLSH:Time}
\end{figure*}

  
\begin{figure*}
  \captionsetup[subfloat]{farskip=0pt}  \center
  \subfloat[\MLtwenty]{    \includegraphics[width=0.23\linewidth]{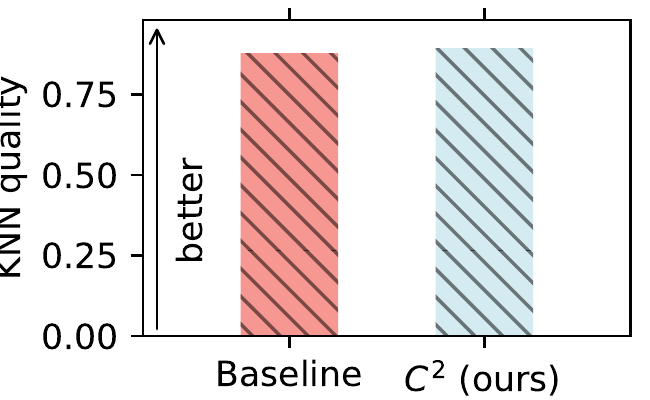}}\hfil
  \subfloat[\AM]{
    \includegraphics[width=0.23\linewidth]{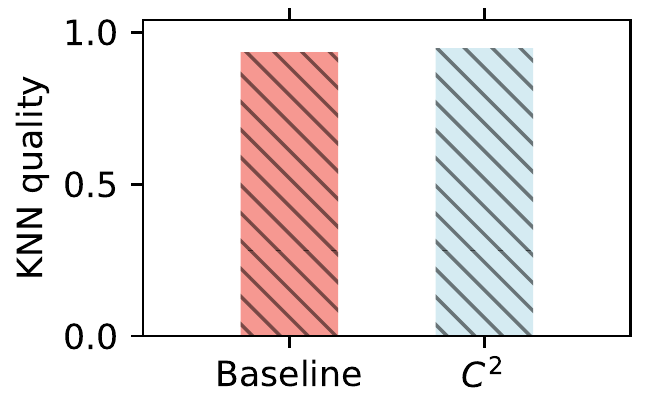}}\hfil
  \subfloat[\DBLP]{
    \includegraphics[width=0.23\linewidth]{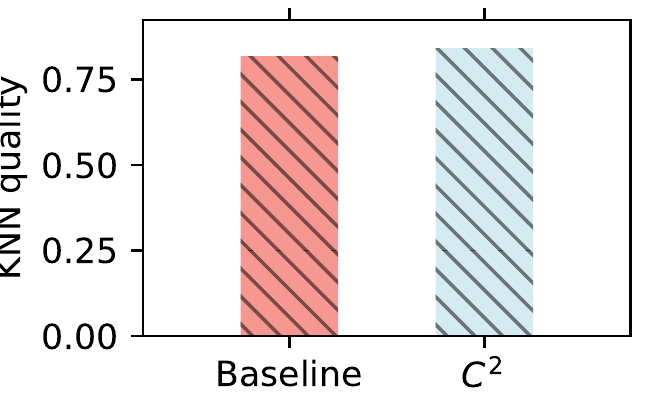}}\hfil
  \subfloat[\GW]{
    \includegraphics[width=0.23\linewidth]{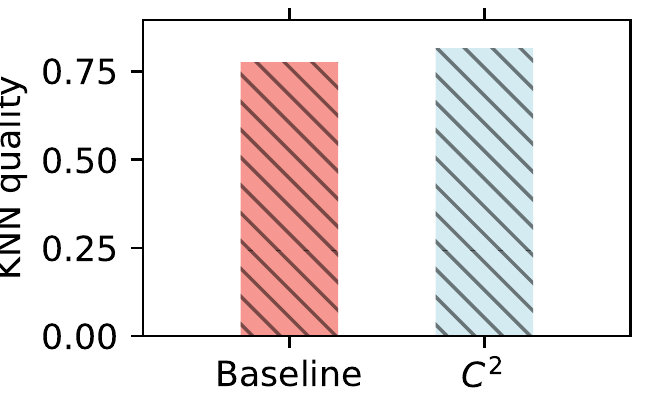}}
  \caption{KNN quality using \hyperLSH and the fastest approach for each dataset (higher is better). Baseline refers to the fastest native approach. On the four above datasets, \HyperLSH (\HLSH) provides a slightly improved KNN quality.}
  \label{fig:HyperLSH:KNNquality}
\end{figure*}






\newcommand{\myHruleReco}{\cline{1-4}}

\begin{table}[tb]
  \caption{Recommendation quality using the brute force approach and Cluster-and-Conquer ($C^2$) while recommending 30 items to every user. The loss in recall ($-2.05\%$ on average) incurred by Cluster-and-Conquer is small.
  }
  \label{tab:recall}
  \renewcommand{\mycolsep}{\hspace{0.6em}}
  \centering
  \small
    \begin{tabular}{@{}l@{\mycolsep}c@{\mycolsep}c@{\mycolsep}c@{}}
      Dataset & Brute force & \bf \HyperLSH & \bf $\Delta$ \\ 

      \myHruleReco

      \MLone & $0.218$ & $0.214$ & $-0.004$ \\ 
      \MLten & $0.273$ & $0.271$ & $-0.003$ \\ 
      \MLtwenty & $0.256$ & $0.253$ & $-0.003$ \\ 
      \AM & $0.595$ & $0.570$ & $-0.025$ \\ 
      \DBLP & $0.360$ & $0.355$ & $-0.005$ \\ 
      \GW & $0.268$ & $0.261$ & $-0.007$ \\ 

      \myHruleReco
      
    \end{tabular}
\end{table}

\subsection{\HyperLSH in action}
\label{subsec:HyperLSH:reco}

We study the practical impact of the approximate KNN quality of \hyperLSH on the iconic item recommendation problem.
We use a simple collaborative filtering procedure, and compare the recommendations obtained with exact KNN graphs to recommendations obtained with \hyperLSH.
Table~\ref{tab:recall} shows the recall obtained when recommending $30$ items to each user in \MLone, \MLten, \MLtwenty, \AM, \DBLP, and \GW using 5-fold cross-validation.
The results of the exact KNN graphs are labeled BruteForce while the ones obtained with \hyperLSH are labeled \HLSH.
The loss in recall is small: we obtain an average loss of $2.05\%$, with loss values ranging from $1.10\%$ on \MLten to $4.20\%$ on \AM.
\hyperLSH provides KNN graphs that are good enough to perform item recommendation with almost no loss, demonstrating its practical potential to accelerate end-user applications with close to no impact.

\renewcommand{\myHrule}{\cline{3-8}}

\begin{table}[tb]
  \caption{Impact of the use of \FMH functions (FRH for short) on 
    \HyperLSH. The gain and speed-up values are computed against the best baseline of Table~\ref{tab:time_KNNquality_HyperLSH}. \FMH functions provide important speeds-up at the cost of a small loss in KNN quality.}
  \label{tab:time_KNNquality_impact_HyperLSH}
  \renewcommand*{\arraystretch}{1.2}
  \centering
  \small
  \newcommand{\datasetTwo}[1]{\multirow{2}{*}{\begin{sideways}#1\end{sideways}}}
  \resizebox{0.48\textwidth}{!}{
  \begin{tabular}{@{\hspace{-0.5em}}r@{\hspace{-0.2em}}cr@{\hspace{0.7em}}c@{\mycolsep}r@{\mycolsep\mycolsep}c@{\mycolsep}c@{}c}
    & & Mechanisms & \bf {\bf time (s)} & {\bf Gain} & \hspace{-0.5em}{\bf Speed-up}  & \bf Quality & $\Delta$  \\
      \myHrule
      \vertbrace{5}{\textit{datasets}}{5.5em}{0.25em}
      
      & \datasetTwo{\mlten} & 
      \MinHash	& $126.74$	 &	$-15.24\%$&$\times 0.87$~	& $0.93$	& $+0.03$ \\
      & & \high FRH (ours)	& \high$\bf 27.79$	 &	\high $74.73\%$&\high $\times 3.96$~	& \high$0.89$	& \high$-0.01$ \\
      
      \datasetsep
      
      & \datasetTwo{\am} & 
      \MinHash 	& $97.31$	 & $-55.90\%$&$\times 0.64$~		& $0.95$	& $+0.02$ \\
      & & \high FRH (ours)   	& \high$\bf 14.11$	 & \high$77.39\%$&\high $ \times 4.42$~		& \high$0.95$	& \high$+0.02$ \\
      \myHrule

      
      

      
      

      
      
      
    \end{tabular}
}
\end{table}

\subsection{Impact of \FMH}
\label{sec:HyperLSH:impact}


\HyperLSH combines several key mechanisms: \FMH functions and their recursive splitting mechanism, the independent KNN computations, and the use of \GF. 
To investigate the impact of \FMH on the overall approach,
we replace \FMH with \MinHash in \hyperLSH.
\MinHash functions are classically employed in the \LSH algorithm and create one cluster per item by design.
In the \hyperLSH/\MinHash variant, we use $t$ \MinHash functions to create $t \times m$ clusters, without splitting.
The local KNN graphs are computed independently using \GF on the $t \times m$ clusters, then merged as in \hyperLSH.

Table~\ref{tab:time_KNNquality_impact_HyperLSH} summarizes the results of all corresponding experiments.
  The table shows \FMH has a decisive impact on the performance of our approach: it decreases the computation time by $78.07\%$ (\MLten) and $85.50\%$ (\AM) while providing a competitive quality.

These results also demonstrate that \HYPERLSH works well on both sparse and dense datasets, \MLten and \AM being representative of both categories (see Sec.~\ref{sec:datasets}).

\renewcommand{\myHrule}{\cline{3-8}}

\begin{table}[tb]
  \caption{Impact of the use of \GF functions (Golfi for short) on 
    \HyperLSH. Gain and speed-up against the best baseline of Table~\ref{tab:time_KNNquality_HyperLSH} (with \GF). \GF provides an important speed-up to \hyperLSH, which remains nevertheless competitive even on raw data.
  }%
  \label{tab:time_KNNquality_impact_gf_HyperLSH}
  \renewcommand*{\arraystretch}{1.2}
  \centering
  \small
  \newcommand{\datasetTwo}[1]{\multirow{2}{*}{\begin{sideways}#1\end{sideways}}}
  \resizebox{0.48\textwidth}{!}{
    \begin{tabular}{@{\hspace{-0.5em}}r@{\hspace{-0.2em}}cr@{\hspace{0.7em}}c@{\mycolsep}r@{\mycolsep\mycolsep}c@{\mycolsep}c@{}c}
    
      & & Mechanisms & \bf {\bf time (s)} & {\bf Gain} & \hspace{-0.5em}{\bf Speed-up}  & \bf Quality & $\Delta$  \\
      \myHrule
      \vertbrace{5}{\textit{datasets}}{5.5em}{0.25em}
      
      & \datasetTwo{\mlten} & 
      Raw data	& $111.29$	 &	$-1.19\%$&$\times 0.99$~	& $0.94$	& $+0.04$ \\
      & & \high Golfi (ours)	& \high$\bf 27.79$	 &	\high $74.73\%$&\high $\times 3.96$~	& \high$0.89$	& \high$-0.01$ \\
      
      \datasetsep
      
      & \datasetTwo{\am} & 
      Raw data 	& $35.05$	 & $43.84\%$&$\times 1.78$~		& $0.95$	& $+0.02$ \\
      & & \high Golfi  (ours)   	& \high$\bf 14.11$	 & \high$77.39\%$&\high $ \times 4.42$~		& \high$0.95$	& \high$+0.02$ \\
      \myHrule

      
      

    \end{tabular}
}
\end{table}

\subsection{Impact of \GF}
\label{sec:impact-gf}

We now study the impact of \GF on \hyperLSH by computing the KNN graphs both with \GF and with the raw data.
The results of these experiments are summarized in Table~\ref{tab:time_KNNquality_impact_gf_HyperLSH}.
The gains and speed-ups are computed against the baselines of Table~\ref{tab:time_KNNquality_HyperLSH} (which use \GF).
Despite this disadvantage, \hyperLSH without \GF is still competitive against the baselines, producing a $43.84\%$ gain in computation time on \AM, and being only slightly slower than the best competitor on \MLten.

\commentFT{Dropped last sentence, as repeat of what said earlier.}


\section{Sensitivity Analysis of Key Parameters}
\label{sec:HyperLSH:sensitivity}
\commentFT{I think we can strongly reduce this section, which is nice to have, but not core. I'd focus only on \AM in Sec.~\ref{subsec:HyperLSH:nbhashnbclusters}, with only one figure (Figure~\ref{fig:HyperLSH:nb_hash_AM}?), and only keeping Figure~\ref{subsec:HyperLSH:nbhashtry} in Sec.~\ref{subsec:HyperLSH:nbhashtry}.}
The performances of \hyperLSH depend on many parameters: the number of clusters per hash function, the number of hash functions, and the maximum size of the clusters. 
In this section, we study the influence of each of these parameters. 
Unless stated otherwise, the parameters are the same as in the previous sections: more specifically the number of clusters per hash function $b$ is $4096$, the number of hash functions $t$ is $8$, and the maximum size of the clusters $N$ is $2000$. 

\newcommand{\graphScale}{0.70}

\begin{figure}[tb]
  \centering
  \subfloat[\MLten]{
    \includegraphics[scale=\graphScale]{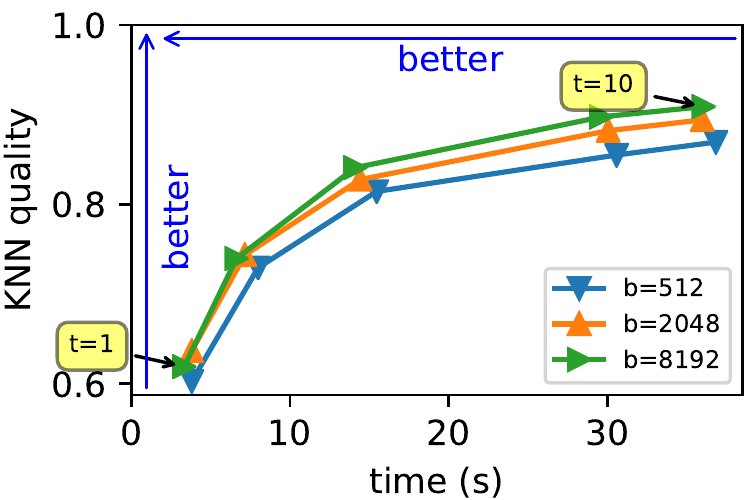}
    \label{fig:HyperLSH:nb_hash_ml10}
  }
  \\
  \subfloat[\AM]{
    \includegraphics[scale=\graphScale]{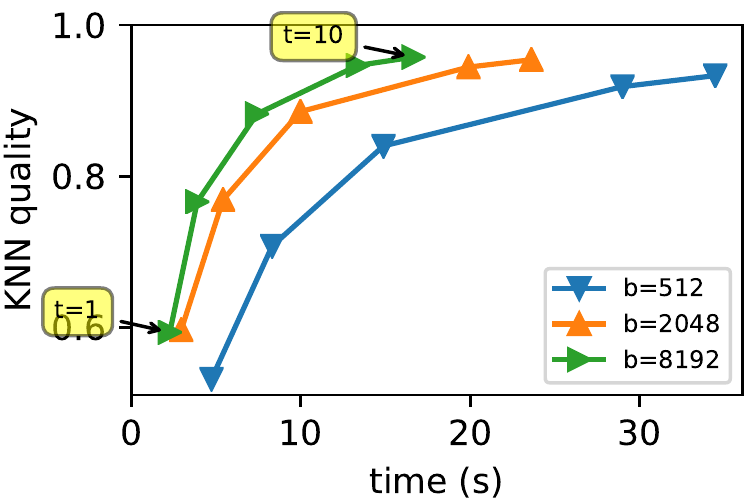}
    \label{fig:HyperLSH:nb_hash_am}
  }
  \caption{Effect of the number of hash functions $t$ on \hyperLSH, for \MLten and \AM.
    Each curve shows the impact of $t$ as it takes the values $\{1,2,4,8,10\}$ for a given number of clusters $b\in \{512,2048,8192\}$.
    A higher $t$ trades off time for quality, but values beyond $8$ offer diminishing returns.}
  \label{fig:HyperLSH:nb_hash}
\end{figure}

\subsection{Number of clusters and hash functions}
\label{subsec:HyperLSH:nbhashnbclusters}

Figure~\ref{fig:HyperLSH:nb_hash} charts how \hyperLSH performs for three different values of $b$ ($512$, $2048$ and $8192$), and five values of $t$ ($1$, $2$, $4$, $8$, and $10$) on two time$\times$quality plots.
Each curve corresponds to a given value of $b$, with the points of the curve obtained by varying $t$.
The figure shows that $t$, the number of hash functions, captures a trade-off between computation time and KNN quality: more hash functions improve quality, but at the cost of higher computation times, and with diminishing returns beyond $t=8$.\commentFT{I've removed the logarithmic comment, as a log formally grows indefinitely while quality here is necessarily capped. I've also tighten up the text which was a bit long.}

In contrast to $t$, increasing $b$, the number of clusters per hash function,\commentFT{adding this precision, as splitting increases the final number of buckets, and there are $b\times t$ cluster before splitting} improves both the computation time and the KNN quality (albeit at the cost of a higher memory consumption). The impact of $b$ is more pronounced on \AM than on \MLten. As we will see in the next section, this is probably because 
recursive splitting strongly impacts \MLten and limits the influence of $b$ by adding a large number of additional clusters. By contrast, recursive splitting has no impact on \AM with $N=2000$ (the value used in Fig.~\ref{fig:HyperLSH:nb_hash}), with the effect that the final number of clusters per hash function is solely determined by $b$.

\begin{figure*}[tb]
  \renewcommand{\graphScale}{0.75}
  \centering
  \parbox{0.32\textwidth}{%
    \vspace{1em}
    
\includegraphics[scale=\graphScale]{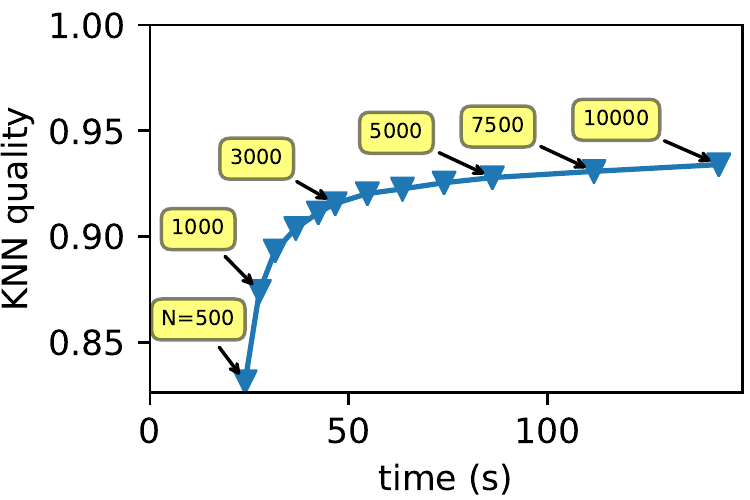}
  \vspace{0.1em}
  \caption{Effect of the maximum cluster size $N$ on the KNN quality and computing time of \hyperLSH on \MLten. Reducing $N$ improves time at the expense of quality.\commentFT{To check, the initial caption said the contrary: increasing $N$ decreases the computation time.}}
  \label{fig:HyperLSH:nb_hash_try:mlten}%
}
  \hfill{}
\renewcommand{\graphScale}{0.75}
\parbox{0.65\textwidth}{%
  \centering%
  \subfloat[\MLten]{%
    \includegraphics[scale=\graphScale]{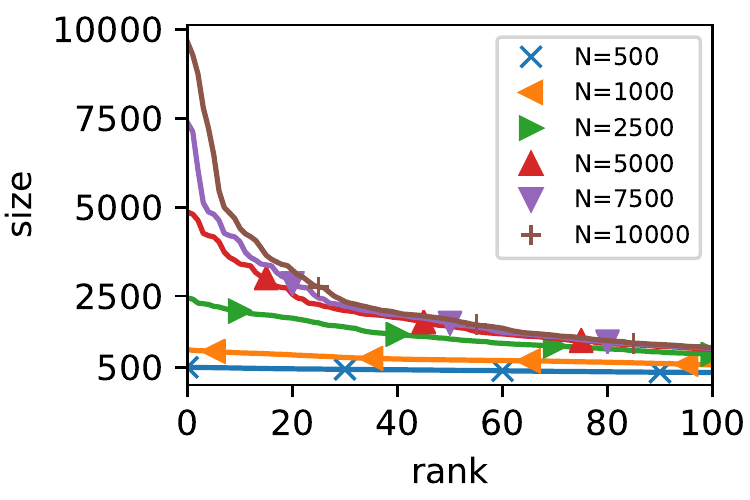}
    \label{fig:HyperLSH:rank_size:mlten}
  }
  \subfloat[\AM]{%
    \includegraphics[scale=\graphScale]{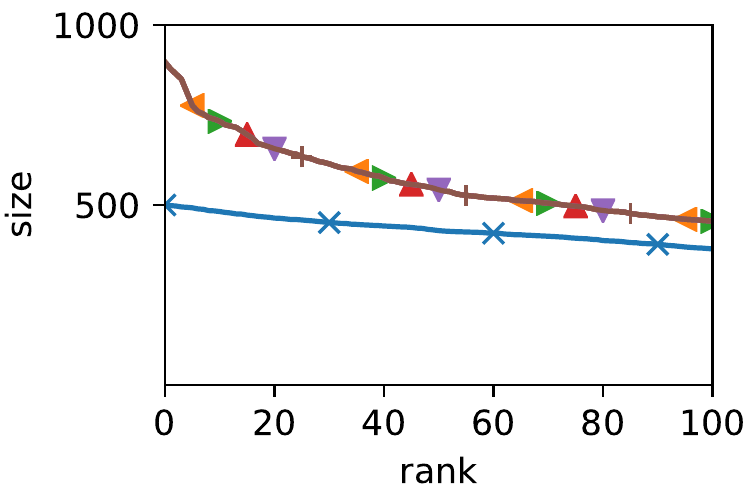}
    \label{fig:HyperLSH:rank_size:am}
  }
  \caption{Effect of the maximum cluster size $N$ on the $100$ biggest clusters of \hyperLSH for \MLten and \AM. On \MLten, the biggest clusters have a size close to $N$ while $N\geq 1000$ has no impact on \am for which the biggest clusters have a size lower than $1000$.}
  \label{fig:HyperLSH:rank_size}%
}
\end{figure*}

\subsection{The recursive splitting strategy}
\label{subsec:HyperLSH:nbhashtry}


Figure~\ref{fig:HyperLSH:nb_hash_try:mlten} shows the
impact of the maximum cluster size $\nbHashTry$ on the computation time and KNN quality of \HyperLSH on \MLten when
$\nbHashTry$ varies from $500$ to $10,000$.
On \MLten increasing $\nbHashTry$ leads to a higher KNN quality, but at the cost of a longer computation time, with a kneepoint around $N=3000$.
By contrast, \AM shows almost no variation for the same values of $N$, and the corresponding plot is omitted for space reasons. 


The difference in the impact of $\nbHashTry$ between \MLten and \AM can be traced back to the popularity distribution of items in each dataset, which in turn impacts how the recursive splitting procedure affects each of them.
This is illustrated in Figure~\ref{fig:HyperLSH:rank_size}, which shows the size of the $100$ biggest clusters in \MLten and \AM for different values of $\nbHashTry$ ranging from $500$ to $10000$. In \MLten (Fig.~\ref{fig:HyperLSH:rank_size:mlten}), raw clusters (without splitting) are highly unbalanced (which is visible for high values of $N$ in the figure). As $N$ decreases, the size of the resulting clusters becomes more uniform, reducing computation times, but scattering similar users in distinct clusters, and thus hurting quality.

By contrast, the largest raw cluster in \AM (Fig.~\ref{fig:HyperLSH:rank_size:am}) contains fewer than $1000$ users. As a result, except for the smallest value of $N=500$, \HyperLSH on \AM does not use recursive splitting and is immune to the impact of $N$.

\section{Related Work}
\label{chap:HyperLSH:related}

KNN graphs are a key mechanism in many problems ranging from classification~\cite{gorai2011employing,DBLP:journals/corr/NodarakisSTTP14} to item recommendation~\cite{Hyrec,Levandoski:2012:LLR:2310257.2310356,linden2003amazon}.
Also, KNN graphs are the first step of more advanced machine-learning techniques~\cite{chen2009fast}.

In small dimension, i.e. when the item set is small, the computation of an exact KNN graph is achieved with specialized data structures~\cite{bentley1975multidimensional,beygelzimer2006cover,liu2004investigation}. 
In high dimension, these techniques are more expensive than the brute force approach.
Computing an exact KNN graph efficiently  in high dimension remains an open problem.

To speed-up the computation of the KNN graph in high dimension, recent approaches decrease the number of similarities computed.
In LSH~\cite{LSH1,LSH2}, each user is placed into several buckets, depending on their profiles.
Two users are in the same bucket with a probability proportional to their similarity.
The KNN of each user is then computed by only considering the users who are in the same buckets as her.
Unfortunately, LSH~\cite{LSH1,LSH2} tends to scale poorly when applied to datasets with large item sets.

Greedy approaches~\cite{Hyrec,NNDescent,bratic2018nn} reduce the number of computing similarities by performing a local search: they assume that neighbors of neighbors are also likely to be neighbors.
They start from a random graph and then iteratively refine each neighborhood by computing similarities among neighbors of neighbors.
These approaches are the most efficient so far on the datasets we are working on.
Still, they spend most of the total computation time computing similarities~\cite{KIFF}.

Compact representations of the data can be used to speed-up similarity computations. 
Several estimators~\cite{dahlgaard2017fast,christiani2018scalable}, including the popular MinHash~\cite{broder1997resemblance,li_theory_2011}, rely on compact datastructures to provide a quick yet accurate estimate of the Jaccard similarity.
Bloom filters~\cite{BloomFilter} can be used as a compact representation of the users' profiles while computing a KNN graph~\cite{alaggan2012blip,gorai2011employing}. Another very simple approach consists in  limiting the size of each user's profile by sampling~\cite{LPeuropar}.
\GF~\cite{GFICDE} is a fast-to-compute compact data structure designed to speed up the Jaccard similarity, which has shown good results across a range of datasets and algorithms~\cite{GFWWW}.


Among the classical clustering techniques, k-means~\cite{macqueen1967some} relies on similarities to provide an efficient clustering:~\cite{xue2005scalable} uses a k-means to cluster the users before computing locally the KNN graph.
Unfortunately, it requires to compute many similarities while our main purpose is to limit as much as possible the number of similarities computed.
On the other hand, LSH~\cite{LSH1,LSH2} clusters users without computing any similarity.
The work presented in~\cite{DBLP:conf/pkdd/ZhangHGL13} uses LSH to cluster users before computing local KNN graphs, as we do. 
The complexity of their clustering approach is $O(n \times m)$, where $n$ is the number of items and $m$ the number of features. 
The performances are good in image datasets, where the dimension $m$, i.e. the number of pixels, is a few thousand but it becomes prohibitive on datasets with much higher dimensions, such as the examples with have considered here.
Similarly, the use of recursive Lanczos bisections~\cite{chen2009fast,lanczos1950iteration} leverages a similar strategy but the divide step has a complexity that makes it more expensive than the brute force approach when used with high dimensional datasets such as the ones we consider.

\section{Conclusions}
\label{chap:HyperLSH:conclusions}

In this paper, we have proposed \hyperLSH, a novel algorithm to compute KNN graphs.
\HyperLSH accelerates the construction of KNN graphs by approximating their graph locality in a fast and robust manner. 
\HyperLSH relies on a divide-and-conquer approach that clusters users, computes locally the KNN graphs in each cluster, and then merges them.
The novelties of the approach are the \FMH functions used to pre-cluster users, their recursive splitting mechanism to produced more balanced clusters, and the fact that the clusters are computed independently, without any synchronization.
Although we have only considered standalone experiments, the general structure of \HyperLSH further makes is particularly amenable to large-scale distributed deployments, in particular within a map-reduce infrastructure.

We extensively evaluated \hyperLSH on real datasets and conducted a sensitivity analysis.
Our results show that \HyperLSH significantly outperforms the best existing approaches, including \LSH, on all datasets, yielding speed-ups ranging from $\times 1.12$ (against \Hyrec on \mlone) up to $\times 4.42$ (against \Hyrec on \am), while incurring only negligible losses in KNN quality.
Finally, we showed that the obtained graphs can replace the exact ones when performing recommendations with almost no discernible impact on recall. 

\bibliographystyle{IEEEtran} 
\bibliography{biblio}





\end{document}